\documentclass[11pt]{article}
\usepackage{mathrsfs}
\usepackage{latexsym}
\usepackage{amsmath,amssymb}
\usepackage[pdftex]{hyperref}
\usepackage{amsthm}
\usepackage{amsfonts,cite}
\usepackage[usenames]{color}
\usepackage{amssymb}
\usepackage{graphicx}
\usepackage{amsmath}
\usepackage{amsfonts}
\usepackage{amsthm}
\usepackage{mathrsfs}
\usepackage{dsfont}
\newcommand\be{\begin{equation}}
\newcommand\ee{\end{equation}}
\newcommand\ber{\begin{eqnarray}}
\newcommand\eer{\end{eqnarray}}
\newcommand\berr{\begin{eqnarray*}}
\newcommand\eerr{\end{eqnarray*}}
\newcommand\bea{\begin{eqnarray}}
\newcommand\eea{\end{eqnarray}}
\newcommand\ba{\begin{array}}
\newcommand\ea{\end{array}}

\newcommand\re{\mathrm{e}}

\newcommand{\bi}{\begin{itemize} }
  \newcommand{\ei}{\end{itemize} }

\newtheoremstyle{mythm}{1.5ex plus 1ex minus .2ex}{1.5ex plus 1ex
minus .2ex}{\kai}{\parindent}{\song\bfseries}{}{1em}{}
\numberwithin{equation}{section}\numberwithin{figure}{section}

\newtheorem{theorem}{Theorem}[section]
\newtheorem{lemma}{Lemma}[section]

\allowdisplaybreaks[4]
\begin{document}
\title{Existence of $U(1)$ Gauged Q-balls for A Field Model with Sixth-order Potential}
\author{Xiaosen Han and Guange Su\\Institute of Contemporary Mathematics\\School of Mathematics and Statistics\\Henan University\\
Kaifeng, Henan 475004, P. R.  China }

\date{}
\maketitle

\begin{abstract}

Q-balls are non-topological solitons in a large family of field theories. We focus on the existence of $U(1)$ gauged Q-balls for a field theory with sixth-order potential. The problem can be reduced to proving the existence of critical points for some indefinite functional. For this, we use a constrained minimization approach to obtain the existence of critical points. Moreover, we establish some qualitative properties of the Q-ball solution, such as monotonicity, boundedness and asymptotic behavior.

\end{abstract}

\medskip

%{\bf PACS numbers.} 02.30.Jr, 02.30.Xx, 11.15.-q, 74.25.Ha
\medskip

%{\bf MSC numbers.} 35J50, 53C43, 81T13

\section{Introduction}

The study of Q-balls has been a subject of interest for several decades. Q-balls are  non-topological solitons arising in complex field theories with a $U(1)$ symmetry  in the 3+1 dimensional Minkowski spacetime\cite{FriedbergLee2,colem85}. It was suggested that such configurations are stable against decay into smaller Q-balls or individual particles \cite{Rosen996}. See a review  \cite{LeePang}, and a recent one \cite{NuagevSh}. Q-balls may exist in models possessing $U(1)$ gauge symmetry and self-interaction potentials \cite{Bishara,Anagno,HongKa,colem85}. Q-balls carry a Noether charge $Q$ which also can be interpreted as the particle number.

Q-balls have been phenomenologically  considered as candidates for macroscopic dark matter \cite{KusenkoSh,KusenkoSt,KusenkoKu,Ponton,Enqvist}. They  may play an important role in primordial phase transition of the early universe \cite{FriedbergLee1,Krylov,Bailu,whitePe},  giving rise to observable gravitational waves \cite{Croon,Bosba,ChibaKa}. Q-balls can be classified into two types, one is called global Q-ball, and the other is named gauged Q-ball. The global Q-ball~only exists in the complex scalar field model with $U(1)$ global symmetry, while the gauged Q-ball~appears in the complex scalar field  model with $U(1)$ gauge (local) symmetry. The gauged~Q-ball~configuration is such that inside the Q-ball the local $U(1)$ symmetry is broken, while  outside the Q-ball the local $U(1)$ symmetry is not. A simple generalization of non-topological solitons from global to gauged $U(1)$~symmetry was firstly  analyzed in the pioneering work by Rosen \cite{Rosen999}. Gauged Q-balls have been studied in numerous works both theoretically and numerically \cite{LoiShn,KinachChop,HeeckRa1, HeeckRa2,PaninSm,NuagevSh,LeeStein,Benci,Gulamov14,LeeYoon,LeviGleiser,ArodzLis,Brihaye,Tamaki,Gulamov15}.

The  purpose  of the present work is to establish an  existence theorem and to obtain  related properties of the gauged Q-balls for a field model with sixth-order potential, proposed in \cite{yu}.  We reduce the problem to the existence of critical points of some indefinite functional. It is not straightforward to directly obtain such critical points. Therefore, to overcome this difficulty, we carry out a constrained  minimization approach, which was first proposed  by \cite{SchecWe},  and extended afterwards in \cite{LinYang,Yang98,GaoYang,ChenGuoDY,Yangbook}. We note that this model has an obvious trivial solution. By estimating the action functional we prove that the constrained minimizer is in fact a non-trivial solution. This provides an effective framework to build the~Q-ball solutions.

Following \cite{yu}, the gauged model which describes the self-interacting complex scalar field $\phi$ minimally interacting with the Abelian gauge field $A_{\mu}$ in Minkowski spacetime $\mathbb{R}^{3,1}$ has the Lagrangian density
\begin{equation}
\mathscr{L}=-\frac{1}{4}F_{\mu\nu}F^{\mu\nu}+(D_{\mu}\phi)^{*}D^{\mu}\phi-V(\phi),\notag
\end{equation}
where $F_{\mu\nu}=\partial_{\mu}A_{\nu}-\partial_{\nu}A_{\mu}$ is the strength  tensor, $D_{\mu}\phi=\partial_{\mu}\phi+ieA_{\mu}\phi$ denotes the covariant derivative, $e$ is the gauge coupling constant and the self-interacting potential of the complex scalar field is
\begin{equation}
V(|\phi|)=m^{2}|\phi|^{2}-\frac{h_{1}}{2}|\phi|^{4}+\frac{h_{2}}{3}|\phi|^{6},\notag
\end{equation}
with   positive self-interaction coupling constants $h_{1}$ and $h_{2}$. The potential $V(|\phi|)$ has a global minimum at $\phi=0$  thus  the parameters must satisfy the inequality $3h_{1}^{2}<16h_{2}m^{2}$.

By varying the action $S =\int\mathscr{L}\mathrm{d}^{3}x\mathrm{d}t$ in the corresponding fields, we obtain the field equations of the model as follows
\begin{align}
D_{\mu}D^{\mu}\phi+m^{2}\phi-h_{1}|\phi|^{2}\phi+h_{2}|\phi|^{4}\phi&=0,\label{m}\\
\partial_{\mu}F^{\mu\nu}&=j^{\nu},\label{n}
\end{align}
where the electromagnetic current density $j^{\nu}=ej_{N}^{\nu}$ is expressed in terms of Noether current
\begin{equation}
j^{\nu}_{N}=i[\phi^{*}D^{\nu}\phi-(D^{\nu}\phi)^{*}\phi].\notag
\end{equation}
The total electric charge of the gauged Q-ball is given by
\ber
Q=4\pi\int_{0}^{\infty}j_{0}(r)r^{2}\mathrm{d}r.
\eer

To obtain  static  solutions, one  uses  the following ansatz
\begin{equation}\label{01}
\phi(x,t)=\frac{f(r)}{\sqrt{2}},~~A^{\mu}(x,t)=\eta^{\mu_{0}}A_{0}(r),
\end{equation}
where $f(r)$ and $A_{0}(r)$ are   real profile functions depending on the radical variable $r$.
Substituting ansatz (\ref{01}) into field equations (\ref{m}) and (\ref{n}), we obtain a system of nonlinear differential equations for the profile functions $f(r)$ and $g(r)=-eA_{0}(r)$:
\begin{align}
f''+\frac{2}{r}f'-(m^{2}-g^{2})f+\frac{h_{1}}{2}f^{3}-\frac{h_{2}}{4}f^{5}&=0,\label{a1}\\
g''+\frac{2}{r}g'-e^{2}gf^{2}&=0.\label{a2}
\end{align}
The corresponding expression for the energy in terms of the profile functions can be written as
\begin{equation}\label{a3}
E(f,g)=4\pi\int^{\infty}_{0}\frac{1}{2}\left\{(f')^2+\frac{1}{e^{2}}(g')^2+g^{2}f^{2}+m^{2}f^{2}-\frac{h_{1}}{4}f^{4}+\frac{h_{2}}{12}f^{6}\right\}r^{2}\mathrm{d}r.\notag
\end{equation}
With (\ref{01}), the electric charge $Q$ can be written as
\ber
Q=4\pi e\int_{0}^{\infty}gf^{2}r^{2}\mathrm{d}r.
\eer
The regularity of the Q-ball field configurations and the finiteness of the Q-ball's energy lead to the boundary conditions for the profile functions
\begin{equation}\label{a4}
f'(0)=0, g'(0)=0;
\end{equation}
\begin{equation}\label{a5}
 f(\infty)=0, g(\infty)=g_{\infty},
\end{equation}
where $g_{\infty}\neq0$ is a constant.  For the case $g_\infty<0$, it can be treated similarly as the case  $g_\infty>0$. Therefore, we just need to handle the case $g_{\infty}>0$ in this paper.

Our existence results regarding gauged Q-balls governed by the system of nonlinear differential equations (\ref{a1})-(\ref{a2}) with the boundary condition  (\ref{a4})-(\ref{a5}) may be stated as follows.

\begin{theorem}\label{THM1}
Assume the parameters satisfying $e, m,  h_1, h_2>0$  and
\begin{equation} \label{aa5}
0<g_{\infty}<m, ~~~h_{1}^{2}<\frac{16}{3}h_{2}(m^{2}-g_{\infty}^{2}).
\end{equation}
Consider the system of nonlinear differential equations (\ref{a1})-(\ref{a2}) subject to the boundary conditions (\ref{a4})-(\ref{a5}) governing a pair of profile functions $f(r)$ and $g(r)$ describing the self-interacting complex scalar field $\phi$ and Abelian gauge field $A_{\mu}$ represented by the ansatz (\ref{01}). Then the boundary value problem  always has a positive solution $(f,g)$, which satisfies the following.
\begin{itemize}
\item[(i)] The profile functions $f$ and $g$ are bounded on $(0,\infty)$, more precisely,
 \begin{equation}
    0<f(r)<\sqrt{\frac{2h_{1}}{{h_{2}}}}, ~~0<g(r)<g_{\infty}. \notag
    \end{equation}
\item[(ii)] The profile function $g$ is strictly increasing for all $r>0$.
\item[(iii)] Near $r=0$, $f'$ and $g'$ satisfy the following asymptotic estimates
\ber
f'(r)=O(r),~g'(r)=O(r).\notag
\eer
\item[(iv)] There hold the asymptotic estimates as $r\rightarrow\infty$
    \begin{equation}
   f(r)=O\left(r^{-1}\exp{(-\sqrt{m^{2}-g^{2}_{\infty}}(1-\varepsilon)r})\right),~g(r)=g_{\infty}+O(r^{-1}),\notag
    \end{equation}
where $0<\varepsilon<1$ is arbitrary.
\item[(v)] The electric charge $Q$ given by the integral
\ber
Q=4\pi e\int_{0}^{\infty}gf^{2}r^{2}\mathrm{d}r
\eer
depends on $g_{\infty}$ and approaches zero as $g_{\infty}\rightarrow0$.
\end{itemize}
\end{theorem}

This existence theorem rigorously confirms the numerical results obtained in \cite{yu}.

The rest of this paper is organized as follows. In Section 2, we introduce a suitable constraint and the admissible space,  which enable us to formulate the problem of existence of gauged Q-balls as a constrained minimization problem. We then, in Section 3, find the condition under which the indefinite action functional becomes coercive. Furthermore, we obtain a minimizer of this constrained minimization problem. In Section 4, we show that  the minimizer is indeed a solution of the system of nonlinear differential equations (\ref{a1})-(\ref{a2}) subject to the boundary condition (\ref{a5}). In Section 5, we verify the remaining boundary condition and obtain some qualitative properties of the solution, such as monotonicity, boundedness, asymptotic estimates. Combining the results obtained in Section 2-5, we establish Theorem \ref{THM1}.

\section{Admissible space}

In this section we introduce a suitable  constraint and  the ensuing admissible space   to formulate the problem as a constrained minimization problem.

The difficulty  lies in the fact  that  (\ref{a1})-(\ref{a2}) is not the Euler-lagrange equation of the positive definite energy $E(f,g)$ but the indefinite action functional
\begin{equation}{\label{a6}}
I(f,g)=\frac{1}{2}\int^{\infty}_{0}\left\{(f')^2-\frac{1}{e^{2}}(g')^2-g^{2}f^{2}+m^{2}f^{2}-\frac{h_{1}}{4}f^{4}+\frac{h_{2}}{12}f^{6}\right\}r^{2}\mathrm{d}r.
\end{equation}

To control the negative terms in (\ref{a6}), we are to `freeze' the unknown $g$ by using the method developed in \cite{SchecWe,LinYang,GaoYang,ChenGuoDY}. The admissible space $\mathscr{A}$ for our problem should be defined by
\begin{equation}
\begin{aligned}
\mathscr{A}=&\left\{(f,g)|\text{the functions }f, g\text{ are absolutely continuous on any compact subinterval of } (0,\infty),\right.\\
&\left.E(f,g)<\infty  \text{ and } f(\infty)=0, g(\infty)=g_{\infty}\right\}. \notag
\end{aligned}
\end{equation}
For convenience, we rewrite the action functional (\ref{a6}) in the form
\begin{equation}
I(f,g)=K(f)-J_{f}(g),
\end{equation}
with
\begin{align}
K(f)&=\frac{1}{2}\int^{\infty}_{0}\left\{(f')^2+m^{2}f^{2}-\frac{h_{1}}{4}f^{4}+\frac{h_{2}}{12}f^{6}\right\}r^{2}\mathrm{d}r,\label{a16}\\
J_{f}(g)&=\frac{1}{2}\int^{\infty}_{0}\left\{\frac{1}{e^{2}}(g')^2+g^{2}f^{2}\right\}r^{2}\mathrm{d}r.\label{a9}
\end{align}
 It is difficult to obtain a critical point of $I$ directly in the space $\mathscr{A}$.  Thus, it is crucial to  construct a suitable   constraint  to restrict the indefinite action functional (\ref{a6}) over a smaller admissible space. To proceed further,  we propose the following constraint
\begin{equation}\label{a10}
\int^{\infty}_{0}\left\{\frac{1}{e^{2}}g'\tilde{g}'+g\tilde{g}f^{2}\right\}r^{2}\mathrm{d}r=0,
\end{equation}
for all $\tilde{g}$ such that $\tilde{g}(\infty)=0$ and $J_{f}(g+\tilde{g})<\infty$.

In fact, (\ref{a10}) follows from
$$
\frac{d}{dt}J_{f}(g+t\tilde{g})|_{t=0}=0. \notag
$$
According to the above observation, we now define the constrained admissible space
\begin{equation}\label{a12}
\mathscr{C}=\left\{(f,g)\in\mathscr{A}|(f,g) \text{ satisfies }(\ref{a10})\right\}.
\end{equation}

The first thing we need to do is to   show that $\mathscr{C}$ is nonempty.
\begin{lemma}\label{l0}
$\mathscr{C}\neq\emptyset$.
\end{lemma}
\begin{proof}
Now given $f$ satisfying $K(f)<\infty$, $f(\infty)=0$, we consider the minimization problem
\begin{equation}\label{a13}
\min\{J_{f}(g)|~g(\infty)=g_{\infty}\}.
\end{equation}

Let $\{g_{n}\}$ be a minimizing sequence of $J_{f}$. Since the functional $J_{f}$ defined in (\ref{a9}) is even and for any function $g$ its distributional derivative must satisfy $||g|'|\leq|g'|$, we see that  $J_{f}(g_{n})\geq J_{f}(|g_{n}|)$.  In other words, we may assume that the minimizing sequence $\{g_{n}\}$ is nonnegative, namely, $g_{n}\geq0$ for all $n$. Since
\begin{equation}\label{a14}
|g_{n}(r)-g_{\infty}|\leq\int^{\infty}_{r}|g'_{n}(s)|\mathrm{d}s\leq\left(\int^{\infty}_{r}\frac{1}{s^{2}}\mathrm{d}s\right)^{\frac{1}{2}}\left(\int^{\infty}_{r}g_{n}'^{2}(s){s^{2}}\mathrm{d}s\right)^{\frac{1}{2}}
\leq r^{-\frac{1}{2}}J_{f}^{\frac{1}{2}}(g_{n}),
\end{equation}
we see that $g_{n}(r)\rightarrow g_{\infty}$ uniformly as $r\rightarrow\infty$. Hence $g_{n}\in W^{1,2}_{loc}((0,\infty),r^{2}\mathrm{d}r)$ and for any pair $0<r_{1}< r_{2}<\infty$, $\{g_{n}\}$ is a bounded sequence in $W^{1,2}(r_{1},r_{2})$. Using compact embedding $W^{1,2}(r_{1},r_{2})\rightarrow C[r_{1},r_{2}]$, there exists $g\in W^{1,2}(r_{1},r_{2})$ such that
\begin{align}
&g_{n}\rightarrow g \text{ weakly in } W^{1,2}(r_{1},r_{2}),\label{a7}\\
&g_{n}\rightarrow g \text{ strongly in } C[r_{1},r_{2}],\label{a8}
\end{align}
as $n\rightarrow\infty$. In view of (\ref{a8}), applying Fatou lemma we have
\begin{equation}
J_{f}^{\delta,R}(g)\leq \liminf_{n\rightarrow\infty}J_{f}^{\delta,R}(g_{n})\leq\lim_{n\rightarrow\infty}J_{f}(g_{n}), \notag
\end{equation}
where $J_{f}^{\delta,R}(g)=\frac{1}{2}\int^{R}_{\delta}\left\{\frac{1}{e^{2}}(g')^2+g^{2}f^{2}\right\}r^{2}\mathrm{d}r$. Letting $\delta\rightarrow0$, $R\rightarrow\infty$,
we arrive at
\begin{equation}
J_{f}(g)\leq\lim\limits_{n\rightarrow\infty}J_{f}(g_{n}). \notag
\end{equation}
It is obvious that $g$ solves the minimization problem (\ref{a13}), so that $g$ is a critical point of $J_{f}(g)$ with fixed $f$, of course satisfies (\ref{a10}). Therefore $\mathscr{C}\neq\emptyset$ and the lemma follows.
\end{proof}

Moreover, the convexity of the functional $J_{f}(g)$ about $g$ ensures the uniqueness of  the solutions of (\ref{a13}) for  fixed $f$.

From the above discussion, we see that for any $f$ satisfying $K(f)<\infty$, $f(\infty)=0$, there exists an unique solution of (\ref{a13}).
Thus $g$ depends on $f$ and we may denote $g=g_{f}$, which indicates that $\mathscr{C}$ defined in (\ref{a12}) can be regarded as   the graph of the map $f\mapsto g_{f}$.
Since $g$ is a non-negative critical point of functional $J_{f}(g)$, we obtain from the Euler-lagrange equations of $J_{f}(g)$ that
\begin{equation}
(r^{2}g')'=e^{2}gf^{2}r^{2}\geq0. \notag
\end{equation}
Hence, $r^{2}g'$ is nondecreasing for $r>0$. We claim that
\begin{equation}\label{a11}
\lim_{r\rightarrow0}r^{2}g'=0.
\end{equation}
Suppose that (\ref{a11}) is false. Then there exist two positive constants $c_{0}$ and $r_{0}$ such that
\begin{equation}
|r^{2} g'(r)|\geq c_{0}>0,~0<r<r_{0},\notag
\end{equation}
which implies that
\begin{equation}
\infty=\int^{r_{0}}_{0}\frac{c_{0}^{2}}{r^{2}}\mathrm{d}r\leq\int^{r_{0}}_{0}|g'(r)|^{2}r^{2}\mathrm{d}r
\leq\int^{\infty}_{0}|g'(r)|^{2}r^{2}\mathrm{d}r<\infty.\notag
\end{equation}
This leads to a contradiction. Thus $r^{2}g'\geq0 $ for all  $r>0$. In particular, $g$ is nondecreasing and for all $r\in(0,\infty)$,
\begin{equation}\label{a17}
0\leq g(r)\leq g_{\infty}.
\end{equation}
These properties of $g$ will be used in what follows.

\section{Partial coerciveness and minimization}
In order to find the conditions under which the functional (\ref{a6}) is coercive with respect to $f$ on $\mathscr{C}$, we define
\begin{equation}
\bar{g}(r)=g_{\infty}(1-\re^{-r}).\notag
\end{equation}

Since $g$ is the minimizer of $J_{f}(g)$ with fixed $f$ satisfying $f(\infty)=0$ and $K(f)<\infty$, we have $J_{f}(g)\leq J_{f}(\bar{g})$, which gives us
\begin{equation}\label{b1}
\begin{aligned}
I(f,g)\geq& K(f)-J_{f}(\bar{g})\\
%\geq&\frac{1}{2}\int^{\infty}_{0}\left((f')^2+m^{2}f^{2}-\frac{h_{1}}{4}f^{4}+\frac{h_{2}}{12}f^{6}\right)r^{2}\mathrm{d}r
%-\int_{0}^{\infty}\frac{1}{2}{g_{\infty}}^{2}f^{2}r^{2}\mathrm{d}r\\
%&-\frac{1}{2}\int^{\infty}_{0}\frac{g^{2}_{\infty}}{e^{2}}(1-\exp{(-r)})'^{2}r^{2}\mathrm{d}r\\
%&\geq\frac{1}{2}\int^{\infty}_{0}\left\{(f')^2+m^{2}f^{2}-\frac{h_{1}}{4}f^{4}+\frac{h_{2}}{12}f^{6}\right\}r^{2}\mathrm{d}r
%-\frac{1}{2}\int^{R}_{0}(\varepsilon_{0}+g_{\infty})^{2}f^{2}r^{2}\mathrm{d}r\\
%&-\frac{1}{2}\int^{\infty}_{R}\left\{\frac{1}{e^{2}}(f')^2+(g_{\infty}+\varepsilon_{0})^{2}f^{2}\right\}r^{2}\mathrm{d}r,\\
%&\geq\frac{1}{2}\int^{\infty}_{0}\left\{(f')^2+m^{2}f^{2}-\frac{h_{1}}{4}f^{4}+\frac{h_{2}}{12}f^{6}\right\}r^{2}\mathrm{d}r
%-\frac{1}{2}\int^{\infty}_{0}\left\{\frac{1}{e^{2}}(f')^2+(g_{\infty}+\varepsilon)^{2}f^{2}\right\}r^{2}\mathrm{d}r-c,\\
\geq&\frac{1}{2}\int^{\infty}_{0}\left((f')^2+\left(m^{2}-g_{\infty}^{2}\right)f^{2}-\frac{h_{1}}{4}f^{4}+\frac{h_{2}}{12}f^{6}\right)r^{2}\mathrm{d}r-\frac{g^{2}_{\infty}}{e^{2}}.
\end{aligned}
\end{equation}
Under the assumption (\ref{aa5}), we see that
\ber\label{b2}
m^{2}-g^{2}_{\infty}>0,~~\left(\frac{h_{1}}{4}\right)^{2}-\frac{h_{2}}{3}(m^{2}-g^{2}_{\infty})<0,
\eer
then there exists a positive constant $c$ independent of $f$ such that
\ber\label{b3}
\left((m^{2}-g_{\infty}^{2})-\frac{h_{1}}{4}f^{2}+\frac{h_{2}}{12}f^{4}\right)f^{2}\geq cf^{2}.
\eer

\begin{lemma}
Suppose that (\ref{aa5}) holds. Then we have the partial coerciveness
\begin{equation}\label{b6}
I(f,g)\geq\frac{1}{2}\int_{0}^{\infty}\left((f')^2+cf^{2}\right)r^{2}\mathrm{d}r-\frac{g^{2}_{\infty}}{e^{2}},~(f,g)\in\mathscr{C},
\end{equation}
where $c>0$ is as given in (\ref{b3})
\end{lemma}
\begin{proof}
The lemma follows from  (\ref{b1}) and (\ref{b3}).
\end{proof}

\begin{lemma}\label{l5}
The constrained minimization problem
\begin{equation}\label{b7}
\min\left\{I(f,g)|(f,g)\in\mathscr{C}\right\}
\end{equation}
has a nontrivial solution provided that the condition (\ref{aa5}) is fulfilled.
\end{lemma}
\begin{proof}
Let $\{(f_{n},g_{n})\}$ be a minimizing sequence of $I$. Since the functional $I(f,g)$ is even with respect to $f$, we may assume that $f_{n}\geq0$ for each $n$. From the boundary condition $f(\infty)=0$, we have
\begin{equation}\label{bb}
|f_{n}(r)|\leq\int^{\infty}_{r}|f'_{n}(s)|\mathrm{d}s\leq\left(\int^{\infty}_{r}\frac{1}{s^{2}}\mathrm{d}s\right)^{\frac{1}{2}}\left(\int^{\infty}_{r}f_{n}'^{2}(s){s^{2}}\mathrm{d}s\right)^{\frac{1}{2}}
\leq r^{-\frac{1}{2}}E^{\frac{1}{2}}(f_{n},g_{n}),
\end{equation}
which implies $f_{n}(r)\rightarrow0$ as $r\rightarrow\infty$ uniformly.

In view of (\ref{b6}), we see that $\{f_{n}\}$ is a bounded sequence in $W^{1,2}((0,\infty);r^{2}\mathrm{d}r)$. Thus there exists $f\in W^{1,2}((0,\infty);r^{2}\mathrm{d}r)$ such that
\begin{align}
&f_{n}\rightarrow f \text{ weakly in } W^{1,2}((0,\infty);r^{2}\mathrm{d}r),\label{r1}\\
&f_{n}\rightarrow f \text{ strongly in } C[r_{1},r_{2}],\label{r2}
\end{align}
as $n\rightarrow\infty$ for any $0<r_{1}<r_{2}<\infty$.

In fact, $\{J_{f_{n}}(g_{n})\}$ is a bounded sequence because
\begin{equation}
J_{f_{n}}(g_{n})\leq J_{f_{n}}(\bar{g})\leq g_{\infty}\int^{\infty}_{0}f^{2}_{n}r^{2}\mathrm{d}r+\frac{2g^{2}_{\infty}}{e^{2}}. \notag
\end{equation}
Hence we can find $g\in W^{1,2}(r_{3},r_{4})$ such that
\begin{align}
&g_{n}\rightarrow g \text{ weakly in } W^{1,2}(r_{3},r_{4}),\label{r3}\\
&g_{n}\rightarrow g \text{ strongly in } C[r_{3},r_{4}],\label{r4}
\end{align}
as $n\rightarrow\infty$ for any $0<r_{3}<r_{4}<\infty$.

In the following, we will show that the weak limit $(f,g)$ of the minimizing sequence $\{(f_{n}, g_{n})\}$ %obtained from (\ref{r1}) and (\ref{a7})
actually lies in $\mathscr{C}$. In other words, we need to show that the constraint (\ref{a10}) is preserved in the weak limit as $n\rightarrow\infty$. For this purpose, it suffices to establish the following results
\begin{equation}\label{b8}
\int^{\infty}_{0}(g'_{n}-g')\tilde{g}'r^{2}\mathrm{d}r\rightarrow0,
\end{equation}
\begin{equation}\label{b9}
\int^{\infty}_{0}(g_{n}f^{2}_{n}-gf^{2})\tilde{g}r^{2}\mathrm{d}r\rightarrow0,
\end{equation}
as $n\rightarrow\infty$ for any $\tilde{g}$ satisfying $J_{f}(g+\tilde{g})<\infty$ and $\tilde{g}(\infty)=0$.

From $J_{f}(g+\tilde{g})<\infty$, we have $\tilde{g}'\in L^{2}((0,\infty);r^{2}\mathrm{d}r)$. The weak convergence under global inner products implies that
\begin{equation}
\int^{\infty}_{0}g'_{n}\tilde{g}'r^{2}\mathrm{d}r\rightarrow\int^{\infty}_{0}g'\tilde{g}'r^{2}\mathrm{d}r \notag
\end{equation}
as $n\rightarrow\infty$. This result enables us to arrive at (\ref{b8}). To establish (\ref{b9}), for convenience we set
\begin{align}
\int^{\infty}_{0}(g_{n}f^{2}_{n}-gf^{2})\tilde{g}r^{2}\mathrm{d}r&=\int_{0}^{\delta}(g_{n}f^{2}_{n}-gf^{2})\tilde{g}r^{2}\mathrm{d}r+\int_{\delta}^{R}(g_{n}f^{2}_{n}-gf^{2})\tilde{g}r^{2}\mathrm{d}r
+\int_{R}^{\infty}(g_{n}f^{2}_{n}-gf^{2})\tilde{g}r^{2}\mathrm{d}r\notag\\
&=I_{1}+I_{2}+I_{3}, \notag
\end{align}
where $0<\delta<R<\infty$ are some positive constants.
Similarly, we have
\begin{equation}\label{b10}
|\tilde{g}(r)|\leq\int^{\infty}_{r}|\tilde{g}'(s)|\mathrm{d}s\leq\left(\int^{\infty}_{r}\frac{1}{s^{2}}\mathrm{d}s\right)^{\frac{1}{2}}\left(\int^{\infty}_{r}\tilde{g}'^{2}(s){s^{2}}\mathrm{d}s\right)^{\frac{1}{2}}
\leq cr^{-\frac{1}{2}}
\end{equation}
for any $r>0$. Using (\ref{a14}), (\ref{bb}) and (\ref{b10}), we get
\begin{equation}
g_{n}(r)-g(r), f(r), \tilde{g}(r)=O(r^{-\frac{1}{2}})
\end{equation}
uniformly for any $r>0$. Then, we have in view of (\ref{a17}) that
\ber
|I_{1}|\leq\int_{0}^{\delta}g_{n}|f_{n}^{2}-f^{2}|\tilde{g}r^{2}\mathrm{d}r+\int_{0}^{\delta}|g_{n}-g|f^{2}\tilde{g}r^{2}\mathrm{d}r\leq c(\delta^{\frac{3}{2}}+\delta),\notag
\eer
where $c > 0$ is a constant independent of $n$. It follows that $I_{1}\rightarrow0$ as $\delta\rightarrow0$ uniformly.
Besides, the pointwise convergence results $f_{n}\rightarrow f$, $g_{n}\rightarrow g$ already indicate that $I_{2}\rightarrow0$ as $n\rightarrow\infty$.
To estimate  $I_{3}$, we obtain from (\ref{a17})
\begin{equation}
|I_{3}|\leq c\max_{r\in(R,\infty)}{|\tilde{g}|}\left(\|f_{n}\|^{2}_{L^{2}((R,\infty);r\mathrm{d}r})+\|f\|^{2}_{L^{2}((R,\infty);r\mathrm{d}r})\right),\notag
\end{equation}
which approaches zero uniformly fast as $R\rightarrow\infty$ by using (\ref{b10}). Consequently, (\ref{b9}) follows immediately. Hence the constrained condition (\ref{a10}) is valid for the limiting configuration $(f,g)$ as expected.

In order to control the negative terms on the right-hand side of (\ref{a6}) when we consider the limiting behavior of $I$ over the minimizing sequence $\{(f_{n},g_{n})\}$, we will need the following property
\begin{equation}\label{b12}
g_{n}'\rightarrow g' \text{ strongly in } L^{2}((0,\infty);r^{2}\mathrm{d}r)
\end{equation}
as $n\rightarrow\infty$. We proceed as follows.

The minimizing sequence $\{(f_{n},g_{n})\}$ of course satisfies (\ref{a10}), namely,
\begin{equation}\label{b13}
\int^{\infty}_{0}\left\{\frac{1}{e^{2}}g_{n}'\tilde{g}'+g_{n}\tilde{g}f_{n}^{2}\right\}r^{2}\mathrm{d}r=0.
\end{equation}
Indeed, subtracting (\ref{b13}) from (\ref{a10}) and setting $\tilde{g}=g_{n}-g$ in the resulting expression, we have
\begin{equation}\label{b14}
\int^{\infty}_{0}(g_{n}'-g')^{2}r^{2}\mathrm{d}r=-e^{2}\int^{\infty}_{0}(g_{n}f^{2}_{n}-gf^{2})(g_{n}-g)r^{2}\mathrm{d}r.
\end{equation}
To estimate the right-hand side of (\ref{b14}), we set
\begin{align}
\int^{\infty}_{0}(g_{n}f^{2}_{n}-gf^{2})(g_{n}-g)r^{2}\mathrm{d}r=&\int_{0}^{\delta}(g_{n}f^{2}_{n}-gf^{2})(g_{n}-g)r^{2}\mathrm{d}r
+\int_{\delta}^{R}(g_{n}f^{2}_{n}-gf^{2})(g_{n}-g)r^{2}\mathrm{d}r\notag\\
&+\int_{R}^{\infty}(g_{n}f^{2}_{n}-gf^{2})(g_{n}-g)r^{2}\mathrm{d}r\notag\\
=&\tilde{I}_{1}+\tilde{I}_{2}+\tilde{I}_{3},\notag
\end{align}
for some constants $0<\delta<R<\infty$. Using the uniform estimate (\ref{a17}) and H\"{o}lder inequality, we deduce
\begin{equation}\label{b16}
|\tilde{I}_{1}|\leq c\delta^{\frac{3}{2}}\left(\|f_{n}\|^{2}_{L^{4}((0,\delta);r^{2}\mathrm{d}r)}+\|f\|^{2}_{L^{4}((0,\delta);r^{2}\mathrm{d}r)}\right),
\end{equation}
where $c>0$ is a constant independent of $n$. Since $\{f_{n}\}$ is bounded in $W^{1,2}((0,\infty);r^{2}\mathrm{d}r)$, applying Sobolev embedding $W^{1,2}((0,\infty);r^{2}\mathrm{d}r)\rightarrow L^{p}((0,\infty);r^{2}\mathrm{d}r)$, $p\geq2$, it is clear that $\tilde{I}_{1}$ goes to zero as $\delta\rightarrow0$.
The pointwise convergence results $f_{n}\rightarrow f$, $g_{n}\rightarrow g$ imply $\tilde{I}_{2}\rightarrow0$ as $n\rightarrow\infty$.
For $\tilde{I}_{3}$, we note that, from $\tilde{g}(\infty)=0$, $\tilde{g}$ is bounded near $r=\infty$. Thus, we have
\begin{equation}\label{b18}
|\tilde{I}_{3}|\leq c\max_{(R,\infty)}\left(|g_{n}-g_{\infty}|+|g-g_{\infty}|\right)\int_{0}^{\infty}\left(f^{2}_{n}+f^{2}\right)r^{2}\mathrm{d}r,
\end{equation}
which approaches zero uniformly fast as $R\rightarrow\infty$ by using $(\ref{a14})$. Taking $n\rightarrow\infty$ in (\ref{b14}), we immediately obtain (\ref{b12}).

Applying the compact embedding $W^{1,2}((0,\infty);r^{2}\mathrm{d}r)\rightarrow L^{4}((0,\infty);r^{2}\mathrm{d}r)$, it is seen that
\begin{equation}\label{b19}
\lim_{n\rightarrow\infty}\int^{\infty}_{0}f_{n}^{4}r^{2}\mathrm{d}r=\int^{\infty}_{0}f^{4}r^{2}\mathrm{d}r.
\end{equation}

The term of the form $-f^{2}g^{2}$ will be tackled separately. %and this is another place where the condition () is essential.
We now need the property that
\begin{equation}\label{b20}
\lim_{n\rightarrow\infty}\int^{\infty}_{0} \left\{(g_{n}^{2}-g_{\infty}^{2})f_{n}^{2}\right\}r^{2}\mathrm{d}r= \int^{\infty}_{0} \left\{(g^{2}-g_{\infty}^{2})f^{2}\right\}r^{2}\mathrm{d}r.
\end{equation}
To this end, similarly, we set
\begin{align}
&\int^{\infty}_{0} \left\{(g_{n}^{2}-g_{\infty}^{2})f_{n}^{2}-(g^{2}-g_{\infty}^{2})f^{2}\right\}r^{2}\mathrm{d}r\notag\\
=&\int_{0}^{\delta}\left\{(g_{n}^{2}-g_{\infty}^{2})f_{n}^{2}-(g^{2}-g_{\infty}^{2})f^{2}\right\}r^{2}\mathrm{d}r
+\int_{\delta}^{R}\left\{(g_{n}^{2}-g_{\infty}^{2})f_{n}^{2}-(g^{2}-g_{\infty}^{2})f^{2}\right\}r^{2}\mathrm{d}r\notag\\
&+\int_{R}^{\infty}\left\{(g_{n}^{2}-g_{\infty}^{2})f_{n}^{2}-(g^{2}-g_{\infty}^{2})f^{2}\right\}r^{2}\mathrm{d}r\notag\\
=&\hat{I}_{1}+\hat{I}_{2}+\hat{I}_{3}.\notag
\end{align}
for some $0<\delta<R<\infty$. Similar to (\ref{b16}), there holds the estimate
\begin{equation}
|\hat{I}_{1}|\leq c\delta^{\frac{3}{2}}\left(\|f_{n}\|^{2}_{L^{4}((0,\delta);r^{2}\mathrm{d}r)}+\|f\|^{2}_{L^{4}((0,\delta);r^{2}\mathrm{d}r)}\right),\notag
\end{equation}
which goes to zero uniformly fast as $\delta\rightarrow0$. Obviously, $\hat{I}_{2}\rightarrow0$ as $n\rightarrow\infty$. Since
\begin{equation}
\begin{aligned}
|\hat{I}_{3}|&\leq\int_{R}^{\infty}|g_{n}^{2}-g_{\infty}^{2}|f_{n}^{2}r^{2}\mathrm{d}r+\int_{R}^{\infty}|g^{2}-g_{\infty}^{2}|f^{2}r^{2}\mathrm{d}r\notag\\
&\leq c\max_{(R,\infty)}\left(|g_{n}-g_{\infty}|\right)\int_{R}^{\infty}f_{n}^{2}r^{2}\mathrm{d}r+c\max_{(R,\infty)}\left(|g-g_{\infty}|\right)\int_{R}^{\infty}f^{2}r^{2}\mathrm{d}r,\notag
\end{aligned}
\end{equation}
where $c>0$ is independent of $n$.  In view of (\ref{a14}), we conclude that $\hat{I}_{3}$ goes to zero as $R\rightarrow\infty$ uniformly for $n$. Thus, the property  stated in (\ref{b20}) follows.

Combining (\ref{b13}), (\ref{b19}) and (\ref{b20}), all the negative terms in the functional $I(f_{n},g_{n})$ are under control.
We are now ready to show that the limit configuration $(f,g)$ is a minimizer of the problem (\ref{b7}). To proceed, we rewrite the functional (\ref{a6}) evaluated over the minimizing
sequence $\{(f_{n},g_{n})\}$ as
\begin{equation}\label{b22}
I(f_{n},g_{n})=\frac{1}{2}\int^{\infty}_{0}\left\{f_{n}'^{2}-\frac{1}{e^{2}}g_{n}'^{2}-(g_{n}^{2}-g_{\infty}^{2})f_{n}^{2}
+(m^{2}-g^{2}_{\infty})f_{n}^{2}-\frac{h_{1}}{4}f_{n}^{4}+\frac{h_{2}}{12}f_{n}^{6}\right\}r^{2}\mathrm{d}r.
\end{equation}
Taking $n\rightarrow\infty$ in (\ref{b22}), using (\ref{b12}), (\ref{b19}) and (\ref{b20}), we immediately arrive at the desired conclusion
\begin{equation}
I(f,g)\leq\lim_{n\rightarrow\infty}I(f_{n},g_{n}).
\end{equation}
Therefore, $(f,g)$ is a minimizer of the problem (\ref{b7}).

Note that $(0,g_{\infty})$ is a trivial solution to the mixed boundary value problem (\ref{a1})-(\ref{a5}) with $I(0,g_{\infty})=0$. We must prove that $(0,g_{\infty})$ is not the solution of problem (\ref{b7}). For this purpose, we will show that there exists $(h,g_{h})\in\mathscr{C}$ such that $I(h,g_{h})<0$. Let us consider the function $h:(0,\infty)\rightarrow\mathbb{R}$,
\begin{equation}
h(r)=\left\{
\begin{array}{lll}
 h_{0}r, & & {0<r\leq 1,}\\
   h_{0}, & & {1< r\leq R,}\\
    h_{0}\exp{(R-r)}, & & {r>R,}
    \end{array}\right.\notag
\end{equation}
where $h_{0}$ is a positive constant. Obviously, $K(h)<\infty$, $h(r)\rightarrow0$ as $r\rightarrow\infty$. Using the argument in lemma \ref{l0}, we can find an unique nonnegative function $g_{h}$ satisfying $(h,g_{h})\in\mathscr{C}$.

We now show that $g_{h}$ is strictly increasing on $(0,\infty)$. To see this, we first claim that $g(r)>0$ for any $r\in(0,\infty)$. Otherwise, we may assume that there is a point $r_{0}>0$ such that $g_{h}(r_{0})=0$. Since $r_{0}$ is a minimum point for the function $g_{h}(r)$, we have $g'_{h}(r_{0})=0$. Applying the uniqueness theorem for the initial value problem of ordinary differential equations, we obtain $g_{h}(r)=0$ for all $r\in(0,\infty)$, which contradicts the fact $g_{h}(\infty)>0$. Hence $g_{h}(r)>0$ for any $r\in(0,\infty)$. In view of (\ref{a2}), we have
\begin{equation}{\label{b24}}
(r^{2}g'_{h})'=e^{2}g_{h}h^{2}r^{2}>0,~r\in(0,\infty),
\end{equation}
so that $r^{2}g'_{h}(r)$ is strictly increasing on $(0,\infty)$. Using (\ref{a11}), we immediately obtain $g'_{h}(r)>0$ for all $r\in(0,\infty)$. Therefore, $g_{h}$ is strictly increasing on $(0,\infty)$.

In the following, we will prove $I(h,g_{h})<0$. Integrating equation (\ref{b24}) over $(1,r)$, for any $r\in(1,R)$ we have
\begin{equation}
r^{2}g'_{h}(r)-g'_{h}(1)=e^{2}\int_{1}^{r}g_{h}(s)h^{2}(s)s^{2}\mathrm{d}s>e^{2}h_{0}^{2}g_{h}(1)\int^{r}_{1}s^{2}\mathrm{d}s=\frac{e^{2}h_{0}^{2}g(1)}{3}(r^{3}-1),\notag
\end{equation}
so that
\begin{equation}
g_{h}'(r)>c_{1}r+\frac{c_{2}}{r^{2}},~1<r<R,\notag
\end{equation}
where $c_{1}=\frac{e^{2}h_{0}^{2}g(1)}{3}$ and $c_{2}=g'_{h}(1)-\frac{e^{2}h_{0}^{2}g(1)}{3}$. It follows that
\begin{equation}
\int_{1}^{R}g_{h}'^{2}r^{2}\mathrm{d}r>\int_{1}^{R}\left(c_{1}r+\frac{c_{2}}{r^{2}}\right)^{2}r^{2}\mathrm{d}r=\frac{c_{1}^{2}}{5}R^{5}+O(R^{2}).\notag
\end{equation}
Since $h(r)=h_{0}$ on $(1,R)$, a simple calculation shows that
\begin{equation}
\int_{1}^{R}\left(m^{2}h^{2}-\frac{h_{1}}{4}h^{4}+\frac{h_{2}}{12}h^{6}\right)r^{2}\mathrm{d}r=\frac{1}{3}\left(m^{2}h_{0}^{2}-\frac{h_{1}}{4}h_{0}^{4}+\frac{h_{2}}{12}h_{0}^{6}\right)(R^{3}-1).\notag
\end{equation}
 Consequently, we are led to the following estimate
\begin{equation}
\begin{aligned}
I_{[0,R]}(h,g_{h})&=\frac{1}{2}\int_{0}^{R}\left(h'^{2}-g_{h}'^{2}-\frac{1}{e^{2}}g_{h}^{2}h^{2}+m^{2}h^{2}-\frac{h_{1}}{4}h^{4}+\frac{h_{2}}{12}h^{6}\right)r^{2}\mathrm{d}r\notag\\
&<-\frac{c_{1}^{2}}{5}R^{5}+O(R^{3}).
\end{aligned}
\end{equation}
Besides, it is easily seen that
\begin{equation}
I_{(R,\infty)}(h,g_{h})
\leq\int^{\infty}_{R}\left(h'^{2}+m^{2}h^{2}-\frac{h_{1}}{2}h^{4}+\frac{h_{2}}{12}h^{6}\right)r^{2}\mathrm{d}r=O(R^{2}).\notag
\end{equation}
Therefore, there exists some $R_{1}$ sufficiently large such that for any $R>R_{1}$, $I(h,g_{h})<0$.

Summarizing the above results, we conclude that the function $(f,g)$ obtained as the limit
of the minimizing sequence $\{(f_{n},g_{n})\}$ for the problem (\ref{b7}) satisfies boundary condition (\ref{a5}), $f(r)\geq0$, $g(r)\geq0$
for all $r>0$, $E(f,g)<\infty$, and $I(f,g)<0$. Thus, $(f,g)$ is a nontrivial solution to the constrain minimization problem (\ref{b7}). The proof of lemma \ref{l5} is complete.
\end{proof}

\section{Weak solutions of governing equations}
In this section, we will show that the obtained minimizer is a solution  to equations (\ref{a1}) and (\ref{a2}) subject to partial boundary condition (\ref{a5}).

\begin{lemma}
The nontrivial solution $(f,g)$ to the optimization problem (\ref{b7}) obtained in the last section satisfies the equations (\ref{a1})-(\ref{a2}).
\end{lemma}
\begin{proof}
It is sufficient to show that all the equations (\ref{a1}) and (\ref{a2}) are fulfilled in  the  weak sense. Let $\tilde{f}\in C_{0}^{1}(0,\infty)$ and $t$ be a real parameter confined in a small interval, saying, $|t|<\frac{1}{2}$. For any $|t|<\frac{1}{2}$, there is an unique corresponding function $g_{t}=g_{f+t\tilde{f}}$ such that $(f+t\tilde{f},g_{t})\in\mathscr{C}$. We denote
\begin{equation}
g_{t}=g+\tilde{g}_{t},~~\hat{g}=\left(\frac{d}{dt}\tilde{g}_{t}\right)\bigg|_{t=0}.\notag
\end{equation}

Since $g_{t}(\infty)=g_{\infty}$, $J_{f+t\tilde{f}}(g_{t})<\infty$, we have $\tilde{g}_{t}(\infty)=0$ and $J_{f}(g+\tilde{g}_{t})<\infty$, which means that $\tilde{g}_{t}$ can be considered as a test function in the constrain condition (\ref{a12}). Thus, there hold
\begin{equation}\label{q1}
\int_{0}^{\infty}\left\{g'_{t}\tilde{g}'_{t}+e^{2}g_{t}\tilde{g}_{t}(f+t\tilde{f})^{2}\right\}r^{2}\mathrm{d}r=0,
\end{equation}
\begin{equation}\label{q2}
\int_{0}^{\infty}\left\{g'\tilde{g}'_{t}+e^{2}g\tilde{g}_{t}f^{2}\right\}r^{2}\mathrm{d}r=0.
\end{equation}
Subtracting (\ref{q2}) from (\ref{q1}), we have the relation
\begin{equation}
\int^{\infty}_{0}\left\{\tilde{g}'^{2}_{t}+e^{2}\tilde{g}^{2}_{t}f^{2}\right\}r^{2}\mathrm{d}r=-\int^{\infty}_{0}e^{2}\left\{2tg_{t}\tilde{g}_{t}f\tilde{f}+t^{2}g_{t}\tilde{g}_{t}\tilde{f}^{2}\right\}r^{2}\mathrm{d}r.\notag
\end{equation}
Using the Young inequality, we obtain
\begin{equation}
\int^{\infty}_{0}\left\{\left(\frac{\tilde{g}'_{t}}{t}\right)^{2}+\frac{e^{2}}{2}\left(\frac{\tilde{g}_{t}}{t}\right)^{2}f^{2}\right\}r^{2}\mathrm{d}r
\leq e^{2}\int^{\infty}_{0}\left\{2g^{2}_{t}\tilde{f}^{2}+|g_{t}\tilde{g}_{t}|\tilde{f}^{2}\right\}r^{2}\mathrm{d}r,\notag
\end{equation}
which gives us
\begin{equation}\label{q3}
\int^{\infty}_{0}\left\{\left(\frac{\tilde{g}'_{t}}{t}\right)^{2}+\frac{e^{2}}{2}\left(\frac{\tilde{g}_{t}}{t}\right)^{2}f^{2}\right\}r^{2}\mathrm{d}r\leq c,~~t\neq0,
\end{equation}
where $c>0$ depends on $\tilde{f}$ but independent of $t$ because of the uniform bound
$$
0\leq g,g_{t}\leq g_{\infty}
$$
and the fact that $|\tilde{g}_{t}|\leq|g|+|g_{t}|$. Moreover, we obtain from (\ref{q3})
\begin{equation}\label{q4}
\left|\frac{\tilde{g}_{t}(r)}{t}\right|\leq\int^{\infty}_{r}\left|\frac{\tilde{g}'_{t}(s)}{t}\right|\mathrm{d}s
\leq\left(\int^{\infty}_{r}\frac{1}{s^{2}}\mathrm{d}s\right)^{\frac{1}{2}}\left(\int^{\infty}_{r}\left(\frac{\tilde{g}'_{t}(s)}{t}\right)^{2}s^{2}\mathrm{d}s\right)^{\frac{1}{2}}
\leq cr^{-\frac{1}{2}},~~t\neq0.
\end{equation}
Note that $\tilde{g}_{t}(r)=0$ at $t=0$ for all $r>0$. Taking $t\rightarrow0$ in (\ref{q3}) and (\ref{q4}), we arrive at
$$
J_{f}(\hat{g})<\infty,~~\hat{g}(\infty)=0.
$$
Hence $\hat{g}$ can be used as a test function in (\ref{a10}). Then   the below  equality holds,
\begin{equation}\label{q5}
\int_{0}^{\infty}\left\{g'\hat{g}'+e^{2}g\hat{g}f^{2}\right\}\mathrm{d}r=0.
\end{equation}

Since $(f,g)$ minimizes $I(f,g)$, it is seen that
$$
\left(\frac{d}{dt}I(f+t\tilde{f},g_{t})\right)\bigg|_{t=0}=0.
$$
In view of (\ref{q5}), this expression leads us to get the weak form of (\ref{a1})
\begin{equation}
\int_{0}^{\infty}\left\{f'\tilde{f}'+(m^{2}-g^{2})f\tilde{f}-\frac{h_{1}}{2}f^{3}\tilde{f}+\frac{h_{2}}{4}f^{5}\tilde{f}\right\}r^{2}\mathrm{d}r=0.\notag
%=\int_{0}^{\infty}\left\{\frac{1}{e^{2}}g'\tilde{g}'_{1}+g\tilde{g}_{1}f^{2}\right\}\mathrm{d}r
\end{equation}

Since (\ref{a10}) is the weak form of (\ref{a2}), we find that $(f,g)$ is a weak solution of (\ref{a1})-(\ref{a2}). Using elliptic regularity theory, $(f,g)$ is a classical solution of (\ref{a1})-(\ref{a2}) subject to the boundary conditions $f(\infty)=0$ and $g(\infty)=g_{\infty}$.
\end{proof}

\section{Remaining boundary condition and basic properties}
In this section, we will prove that the solution $(f,g)$ obtained in the last section satisfies the remaining boundary condition (\ref{a4}). As byproducts, we will derive some useful properties of the solution.
\begin{lemma}
For any $m$, $h_{1}$, $h_{2}$ satisfying (\ref{aa5}), let $(f,g)$ be a non-trivial solution of (\ref{a1})-(\ref{a2}) subject to (\ref{a5}). Then
\begin{equation}\label{q30}
0<f<\sqrt{\frac{2h_{1}}{h_{2}}}, ~~0<g<g_{\infty},
\end{equation}
for all $r>0$.
\end{lemma}
\begin{proof}
By virtue of lemma \ref{l5}, we have
$$
f(r)\geq0,~~~0\leq g(r)\leq g_{\infty},~~r>0.
$$
We first claim
\begin{equation}\label{h10}
\liminf_{r\rightarrow0}r^{2}f'(r)=0.
\end{equation}
Suppose otherwise that (\ref{h10}) is not valid. Then there are some $c_{0}>0$ and $r_{0}>0$ such that
\begin{equation}
|r^{2}f'(r)|\geq c_{0},~~r\in(0,r_{0}),\notag
\end{equation}
which leads to
$$
\infty=\int^{r_{0}}_{0}\frac{c^{2}_{0}}{r^{2}}\mathrm{d}r\leq\int^{r_{0}}_{0}(f')^2r^{2}\mathrm{d}r.
$$
This contradicts with $E(f,g)<\infty$. So (\ref{h10}) is valid.

We rewrite equation (\ref{a1}) as following
\begin{equation}\label{h11}
(r^{2}f')'=r^{2}\left((m^{2}-g^{2}(r))f(r)-\frac{h_{1}}{2}f^{3}(r)+\frac{h_{2}}{4}f^{5}(r)\right).
\end{equation}
Because of $m>g_{\infty}$, if $f>\sqrt{\frac{2h_{1}}{h_{2}}}$, we have $(r^{2}f')'>0$ for all $r>0$. From (\ref{h10}), it is obvious that $f'(r)>0$ for any $r>0$, which implies that the boundary condition $f(r)\rightarrow0~(r\rightarrow\infty)$ is not met. Therefore we conclude $0\leq f(r)\leq\sqrt{\frac{2h_{1}}{h_{2}}}$.

We now prove $f(r)>0$. If for some $r_{0}>0$ such that $f(r_{0})= 0$, which indicates that $r_{0}$ is a minimum point of $f(r)$, then $f'(r_{0})=0$. Applying the uniqueness theorem for the initial value problem of ordinary differential equations, we have $f(r)=0$ for all $r\in(0,\infty)$, which contradicts with the statement that $f(r)$ is a non-trivial solution. Similarly, we  have $g(r)>0$.

In order to prove $f<\sqrt{\frac{2h_{1}}{h_{2}}}$, we now assume that there exists a point $r_{1}\in(0,\infty)$ such that $f(r_{1})=\sqrt{\frac{2h_{1}}{h_{2}}}$. Obviously, $r_{1}$ is a  maximum point of $f(r)$, so that $f'(r_{1})=0$, $f''(r_{1})\leq0$. We have, in view of (\ref{a1}),
\begin{equation}
f''(r_{1})=-\frac{2}{r_{1}}f'(r_{1})+(m^{2}-g^{2}(r_{1}))f(r_{1})-\frac{h_{1}}{2}f^{3}(r_{1})+\frac{h_{2}}{4}f^{5}(r_{1})>0,
\end{equation}
which is contradict the fact $f''(r_{1})\leq0$. Thus $f(r)<\sqrt{\frac{2h_{1}}{h_{2}}}$ for all $r\in(0,\infty)$. A similar argument used
in the proof of $f<\sqrt{\frac{2h_{1}}{h_{2}}}$ gives $g<g_{\infty}$.
\end{proof}

\begin{lemma}
 For any $m$, $h_{1}$, $h_{2}$ satisfying (\ref{aa5}), let $(f,g)$ be a non-trivial solution of (\ref{a1})-(\ref{a2}) subject to (\ref{a5}). Then $g$ is strictly increasing on $(0,\infty)$.
\end{lemma}
\begin{proof}
In view of (\ref{q30}), we have
\begin{equation}{\label{c22}}
(r^{2}g')'=e^{2}gf^{2}r^{2}>0,~r\in(0,\infty).
\end{equation}
Thus $r^{2}g'(r)$ is strictly increasing on $(0,\infty)$. From (\ref{a11}), we immediately obtain $g'(r)>0$ for all $r\in(0,\infty)$. Therefore, $g$ is strictly increasing on $(0,\infty)$.

Note that if $g_{\infty}<0$, then $g(r)$ is strictly decreasing and $g_{\infty}<g(r)<0$ for all $r>0$.

\begin{lemma}
For any $m$, $h_{1}$, $h_{2}$ satisfying (\ref{aa5}), let $(f,g)$ be a non-trivial solution of (\ref{a1})-(\ref{a2}) subject to (\ref{a5}). Then
\begin{equation}
\lim_{r\rightarrow0}f'(r)=0,~~\lim_{r\rightarrow0}g'(r)=0.\notag
\end{equation}
In other words, the boundary condition (\ref{a4}) holds. Moreover, $f'$~and $g'$ satisfy the following asymptotic estimates
\ber
f'(r)=O(r),~g'(r)=O(r),~r\rightarrow0.\notag
\eer
\end{lemma}
{\bf Proof.} Integrating (\ref{h11}) over $(0,r)$, we have
\begin{equation}
r^{2}f'(r)=\int_{0}^{r}\left((m^{2}-g^{2}(s))f(s)-\frac{h_{1}}{2}f^{3}(s)+\frac{h_{2}}{4}f^{5}(s)\right)s^{2}\mathrm{d}s,\notag
\end{equation}
for any $r>0$. It follows that
\begin{equation}
|f'(r)|\leq\frac{1}{r^{2}}\left(\int_{0}^{r}(m^{2}-g^{2}(s))f(s)s^{2}\mathrm{d}s+\int_{0}^{r}\frac{h_{1}}{2}f^{3}(s)s^{2}\mathrm{d}s+\int_{0}^{r}\frac{h_{2}}{4}f^{5}(s)s^{2}\mathrm{d}s\right).\notag
\end{equation}
By using (\ref{q30}), it is seen that there exists a positive constant $c$ such that
\begin{equation}
|f'(r)|\leq cr,\notag
\end{equation}
for $r$ near zero. Clearly, $f'(r)\rightarrow0$ as $r\rightarrow0$, and $f'$ satisfies the following asymptotic estimates near $r=0$
\ber
f'(r)=O(r).\notag
\eer

The same method shows that $g'(r)=O(r)$ as $r\rightarrow0$. Obviously, $\lim\limits_{r\rightarrow0}g'(r)=0$. The proof of the lemma is complete.
\end{proof}

\begin{lemma}
For any $m$, $h_{1}$, $h_{2}$ satisfying (\ref{aa5}), let $(f,g)$ be a non-trivial solution of (\ref{a1})-(\ref{a2}) subject to (\ref{a4})-(\ref{a5}). Then there hold the asymptotic estimates
\begin{equation}\label{h15}
f(r)=O\left(r^{-1}\exp{(-\sqrt{m^{2}-g^{2}_{\infty}}(1-\varepsilon)r})\right),~~~~g(r)=g_{\infty}+O(r^{-1})
\end{equation}
for $r\rightarrow\infty$, where $0<\varepsilon<1$ is arbitrary.
\end{lemma}
\begin{proof}
To get decay estimate for $f$, we introduce a new function $\hat{f}=rf$. Clearly, $\hat{f}>0$ for any $r>0$. In view of (\ref{a1}), we have
\begin{equation}
\hat{f}''=(m^{2}-g^{2})\hat{f}-\frac{h_{1}}{2}f^{2}\hat{f}+\frac{h_{2}}{4}f^{4}\hat{f}.\notag
\end{equation}
Define the comparison function
\begin{equation}\label{q20}
\eta(r)=C\exp{(-\sigma(1-\varepsilon)r)},
\end{equation}
where $\sigma=\sqrt{m^{2}-g^{2}_{\infty}}$ and $C>0$ is a constant to be chosen later.
Then,
\begin{equation}\label{qq}
(\hat{f}-\eta)''=\sigma^{2}(1-\varepsilon)^{2}(\hat{f}-\eta)
+\left((m^{2}-g^{2})-\frac{h_{1}}{2}f^{2}+\frac{h_{2}}{4}f^{4}-\sigma^{2}(1-\varepsilon)^{2}\right)\hat{f}.\notag
\end{equation}
Since $g\rightarrow g_{\infty}$ as $r\rightarrow\infty$, we can find a suitably large $r_{\varepsilon}>0$ such that
\begin{equation}
(m^{2}-g^{2})-\frac{h_{1}}{2}f^{2}+\frac{h_{2}}{4}f^{4}-\sigma^{2}(1-\varepsilon)^{2}>0,~~r>r_{\varepsilon}\notag
\end{equation}
which gives us
\begin{equation}\label{q21}
(\hat{f}-\eta)''\geq\sigma^{2}(1-\varepsilon)^{2}(\hat{f}-\eta),
\end{equation}
for any $r>r_{\varepsilon}$. Choose the coefficient $C$ in (\ref{q20}) large enough such that  $(\hat{f}-\eta)(r_{\varepsilon})\leq0$. Furthermore, the finite energy implies that there is a sequence $\{r_{j}\}$, $r_{j}\rightarrow\infty$ as $j\rightarrow\infty$ so that $\hat{f}(r_{j})\rightarrow0$ as $j\rightarrow\infty$. Using this and applying the maximum principle in (\ref{q21}), there holds $\hat{f}<\eta$ for all $r>r_{\varepsilon}$. So the decay estimate for $f$ near infinity stated in (\ref{h15}) is established.

Next we consider the decay estimate for $g$. The equation (\ref{a2}) can be rewritten as
\begin{equation}\label{q22}
\left(r(g(r)-g_{\infty})\right)''=e^{2}g(r)f^{2}(r)r.
\end{equation}
In view of the finite energy condition, there exists a sequence $\{r_{k}\}$, $r_{k}\rightarrow\infty$ as $k\rightarrow\infty$ so that $rg'(r)\rightarrow0$ as $r\rightarrow\infty$. Hence, from (\ref{q22}), we have
\begin{equation}\label{q24}
\left(r(g(r)-g_{\infty})\right)'=-\int_{r}^{\infty}e^{2}g(s)f^{2}(s)s\mathrm{d}s.
\end{equation}
Since $0<g(r)<g_{\infty}$ for any $r>0$, inserting the expression of $f$ into (\ref{q24})
we see that the function $(r[g(r)-g_{\infty}])'$ also vanishes exponentially fast at infinity, so that
\begin{equation}
g(r)=g_{\infty}+O(r^{-1}), ~r\rightarrow\infty\notag
\end{equation}
as expected. Thus the lemma follows.
\end{proof}

Note that the gauged Q-ball does not exist when $|g_{\infty}|>m$ since the asymptotic behavior (\ref{h15}) about $f(r)$ shows oscillating behaviour, leading to an infinite energy.

\begin{lemma}
For any $m$, $h_{1}$, $h_{2}$ satisfying (\ref{aa5}), let $(f,g)$ be a non-trivial solution of (\ref{a1})-(\ref{a2}) subject to (\ref{a4})-(\ref{a5}) with fixed the electric charge
\ber\label{q25}
Q(g_{\infty})=4\pi\int_{0}^{\infty}gf^{2}r^{2}\mathrm{d}r.
\eer
Then $Q$ satisfies the property $Q(g_{\infty})\rightarrow0$ as $g_{\infty}\rightarrow0$.
\end{lemma}
\begin{proof}
Since $f$ vanishes exponentially fast at infinity and $0<g(r)<g_{\infty}$ for all $r>0$, we can apply the dominated convergence theorem to (\ref{q25}) to see that $Q(g_{\infty})\rightarrow0$ as $g_{\infty}\rightarrow0$.
\end{proof}

\end{document}